\documentclass[conference]{IEEEtran}
\usepackage{mathtools}  
\usepackage{amssymb,amsfonts,amsthm}

\DeclareMathOperator{\TV}{TV}
\DeclareMathOperator{\Prov}{Provenance}
\usepackage{algorithmic}
\usepackage{graphicx}
\usepackage{textcomp}
\usepackage{xcolor}
\usepackage{booktabs}
\usepackage[hidelinks]{hyperref}
\usepackage{multirow}
\usepackage{array}
\usepackage{enumitem}
\usepackage{pifont}

\newtheorem{theorem}{Theorem}
\newtheorem{definition}{Definition}

\makeatletter
\renewcommand{\section}{\@startsection{section}{1}{\z@}%
  {-3.5ex \@plus -1ex \@minus -.2ex}%
  {2.3ex \@plus .2ex}%
  {\normalfont\normalsize\bfseries}}
\makeatother

\begin{document}

\title{On the Insecurity of Keystroke-Based AI Authorship Detection: Timing-Forgery Attacks Against Motor-Signal Verification}

\author{
\IEEEauthorblockN{David Condrey}
\IEEEauthorblockA{Writerslogic Inc.}
}

\maketitle

\begin{abstract}
Recent proposals advocate using keystroke timing signals---specifically the coefficient of variation ($\delta$) of inter-keystroke intervals---to distinguish human-composed text from AI-generated content. We demonstrate that this class of defenses is insecure against two practical attack classes: the \emph{copy-type attack}, in which a human transcribes LLM-generated text producing authentic motor signals, and \emph{timing-forgery attacks}, in which automated agents sample inter-keystroke intervals from empirical human distributions. Using 13,000 human typing sessions from the Stony Brook University keystroke corpus and three timing-forgery variants (histogram sampling, statistical impersonation, and generative LSTM), we show that all attacks achieve ${\geq}\,99.8$\% evasion rates against five classifiers trained on seven keystroke features. Despite achieving AUC~$=1.000$ when distinguishing human from fully-automated output, these classifiers classify ${\geq}\,99.8$\% of attack samples as human with mean confidence $\geq 0.993$. We formalize a non-identifiability result under explicit observational assumptions: when the detector observes only keystroke timing, the mutual information between features and content provenance is zero for copy-type attacks. While composition and transcription do produce statistically distinguishable motor patterns (Cohen's $d = 1.28$ within subjects), both conditions yield $\delta$ values 2--4$\times$ above detection thresholds, rendering the distinction security-irrelevant. These systems can confirm a human operated the keyboard, but not whether that human originated the text. Securing content provenance requires fundamentally different architectures---ones that bind the observed writing process to the semantic content of the output.
\end{abstract}


\section{Introduction}

If a student types an essay keystroke-by-keystroke on an instrumented platform, the resulting timing trace looks indistinguishable from genuine composition---even if the student is transcribing a ChatGPT draft from a second screen. This observation exposes a category error in recent AI authorship detection proposals: they conflate \emph{motor presence} with \emph{content origin}. Current approaches to AI text detection fall into two categories: \emph{text-level} methods analyzing linguistic output~\cite{mitchell2023detectgpt,kirchner2023watermark,gehrmann2019gltr}, and \emph{process-level} methods analyzing behavioral signals during composition~\cite{kundu2024,mehta2025}. The latter are motivated by the observation that genuine composition produces characteristic temporal patterns---pauses for thought, bursts of rapid typing, variable rhythm---absent in automated text injection.

Kundu et al.~\cite{kundu2024} proposed using the coefficient of variation of inter-keystroke intervals,
\begin{equation}
\label{eq:delta_def}
\delta \;\coloneqq\; \frac{\sigma_{\mathrm{IKI}}}{\mu_{\mathrm{IKI}}}\,,
\end{equation}
as a discriminator, reporting strong separation between human composition and AI-generated content injected without motor activity. Mehta et al.~\cite{mehta2025} extended this with additional temporal features and TypeNet-style embeddings. Both works implicitly assume that the presence of human motor signals constitutes evidence of human \emph{authorship}---that is, the human composed the content rather than transcribed it. While transcription attacks may appear obvious in hindsight, none of the existing keystroke-based AI detection systems model them in their threat assumptions or evaluate against them.

\textbf{This assumption does not hold under a rational adversary.} We define \emph{content provenance} as the property that the typist cognitively originated the text---as opposed to transcribing, paraphrasing, or dictating content produced by an external source. These systems implicitly treat \emph{motor authenticity} (the signal was produced by a human body) as equivalent to content provenance, but the two are independent: The \emph{copy-type attack}---in which an adversary generates text with an LLM, then physically types it on the instrumented platform---defeats these systems because the resulting motor signal is genuine by construction. Formally: for any feature extractor $f$ operating on keystroke timing $\tau$, character sequence $s$, and the typist's motor model $M_u$,
\[
I\bigl(f(\tau)\,;\,\Prov \mid s,\, M_u\bigr) = 0
\]
(Theorem~1; see Eq.~\ref{eq:nonident}). Even relaxing this idealization, the composition-transcription gap (Cohen's $d \approx 1.28$) is operationally unexploitable at acceptable false-rejection rates (\S\ref{sec:nonident}).

\textbf{Scope and non-goals.} This paper evaluates \emph{timing-only, content-agnostic keystroke-based AI authorship detectors} under adversarial transcription. We do not analyze multimodal systems (gaze, revision graphs), challenge-response schemes, or cryptographic content-binding designs. Our claims apply exclusively to detectors whose security guarantees derive solely from keystroke timing features. This paper should be read as a security limitations and attack paper: we identify a dominant adversary, formalize a non-identifiability boundary, and empirically validate that boundary across existing detectors.

\subsection{Contributions}

\begin{enumerate}[leftmargin=*]
\item We define the \textbf{copy-type attack} and show, under explicit observational assumptions, that it is non-identifiable by any classifier operating solely on keystroke timing (\S\ref{sec:threat}--\ref{sec:nonident}).

\item We evaluate three \textbf{timing-forgery attacks} achieving ${\geq}\,99.8$\% evasion against five classifiers on $n = 13{,}000$ human sessions and 2,000 attack sessions, with all samples exceeding the $\delta$ threshold (\S\ref{sec:experiments}).

\item We perform a \textbf{seven-feature ablation} demonstrating that no individual keystroke feature reliably separates all attack variants from genuine human typing (\S\ref{sec:ablation}).

\item We establish \textbf{non-identifiability bounds} using Jensen-Shannon divergence and the data processing inequality (\S\ref{sec:nonident}).

\item We identify \textbf{semantic coherence} across revision histories as a viable defense direction (\S\ref{sec:defense}).
\end{enumerate}

\section{Background and Related Work}
\label{sec:related}

AI authorship detection methods can be organized into three categories: (1)~\emph{content-based} approaches that analyze the linguistic output itself (perplexity, watermarks, stylometry); (2)~\emph{process-signal} approaches that analyze behavioral traces generated during composition (keystroke timing, mouse dynamics, revision history); and (3)~\emph{identity-verification} approaches that confirm the typist's identity via biometric matching. Keystroke dynamics is well-established for identity verification~\cite{banerjee2012,monrose1997,killourhy2009,typenet2020}, with motor noise as the primary source of inter-keystroke timing differences~\cite{buschek2015}. We analyze the limitations of category (2): process-signal methods cannot establish content provenance because motor signals are agnostic to whether the typist composed or transcribed the text.

\subsection{Attacks on Keystroke Biometrics}

Keystroke timing is spoofable via statistical imitation~\cite{tey2013}, robotic injection~\cite{serwadda2013}, synthetic generation~\cite{monaco2016,eizagirre2022}, and mimicry~\cite{stefan2012} (survey:~\cite{roy2022}). Our contribution extends this to AI authorship detection: rather than impersonating a \emph{different user}, the attacker presents their \emph{own authentic} motor signals while submitting AI-generated content.

\subsection{AI Text Detection}

Text-level methods (perplexity~\cite{mitchell2023detectgpt}, visualization~\cite{gehrmann2019gltr}, watermarks~\cite{kirchner2023watermark}, stylometry~\cite{uchendu2020}) face theoretical limits: paraphrasing defeats any detector as model quality improves~\cite{sadasivan2023}, and systematic bias affects non-native writers~\cite{liang2023}. These motivate the process-level approach we attack.

\subsection{Process-Level AI Detection}

Kundu et al.~\cite{kundu2024} proposed $\delta$ as discriminator ($d \approx 1.28$ between composition and transcription, 1,060 participants); Mehta et al.~\cite{mehta2025} extended this with TypeNet (F1 ${>}\,97$\%); concurrent work applies the approach to Korean~\cite{roh2025korean} and LLM-assisted conditions~\cite{roh2025auth}. All evaluate against paste/API injection (motor-absent) and none considers a copy-type adversary---the attack we formalize.

\subsection{Biometric Presentation Attacks}

Presentation attacks are well-studied (gummy fingers~\cite{matsumoto2002}, face photos~\cite{sharif2016}, voice replay~\cite{chen2021}); ISO/IEC 30107~\cite{iso30107} formalizes PAD in terms of attacker expertise, equipment, and access. Our copy-type attack is a novel class with \emph{minimal} attack potential: rather than spoofing someone else's biometric, the attacker presents their \emph{own authentic} biometric while submitting someone else's content---outside the PAD framework's assumptions.

\subsection{Adversarial Machine Learning Context}

Certified robustness~\cite{cohen2019} and adversarial training~\cite{madry2018} assume a known, constrained perturbation set ($\|x' - x\|_p \leq \epsilon$); the copy-type attack violates this premise entirely. Copy-type does not perturb a legitimate input~\cite{carlini2017,goodfellow2015}; it directly produces samples drawn from the \emph{same} motor-execution distribution as legitimate ones. The distinguishing variable (content origin) is a latent confounder unobserved in the feature space, yielding irreducible Bayes error from distribution overlap rather than a robustness gap exploitable by tighter certificates (\S\ref{sec:nonident}).

\section{Threat Model}
\label{sec:threat}

\subsection{System Model}

We consider a \emph{Motor-Signal Verification System} (MSVS) $\mathcal{V}$ that accepts a document $d$ and its associated keystroke trace $\tau = \{(k_i, t_i)\}_{i=1}^{n}$ and outputs:
\[
\mathcal{V}(d, \tau) \in \{\text{HUMAN}, \text{AI}\}
\]

The system extracts features $\mathbf{f}(\tau) \in \mathbb{R}^p$ from the trace and applies a classifier $C: \mathbb{R}^p \to \{0, 1\}$.

\textbf{Trust assumptions.} The system trusts that: (1)~keystroke events are captured faithfully by the client-side instrumentation (JavaScript keydown/keyup events in web deployments); (2)~the captured trace corresponds to the submitted document; (3)~no middleware intercepts or modifies events. For timing-forgery attacks, assumption (1) is violated. For copy-type attacks, \emph{all} assumptions hold---the attack is indistinguishable by design.

\textbf{Scope.} We restrict attention to systems claiming security guarantees from keystroke timing alone, without auxiliary biometric or cognitive sensors (e.g., eye tracking, gaze entropy, concurrent process monitoring). Systems incorporating such sensors fall outside our threat model and are not affected by our attacks.

\textbf{Features.} We evaluate the following feature set, encompassing all features proposed in~\cite{kundu2024,mehta2025}:
\begin{itemize}[leftmargin=*]
\item $\delta = \mathrm{CV}(\mathrm{IKI})$: coefficient of variation of inter-keystroke intervals
\item $\bar{t}$: mean inter-keystroke interval (ms)
\item $\sigma^2_{\mathrm{IKI}}$: IKI variance
\item $\rho$: pause density (fraction of IKIs $> 500$ms)
\item $\beta$: mean burst length (consecutive IKIs $< 150$ms)
\item $H$: Shannon entropy of IKI distribution (50-bin histogram, range [0, 2000]ms)
\item $\gamma$: digraph variability (std of consecutive IKI differences)
\end{itemize}

\subsection{Adversary Model}

The adversary $\mathcal{A}$ seeks to submit AI-generated text while receiving a HUMAN classification.

\begin{definition}[Copy-Type Attack]
The adversary generates document $d^*$ using LLM $\mathcal{M}$, then physically types $d^*$ character-by-character on the instrumented platform, producing trace $\tau^* = \operatorname{Type}(\mathcal{A},\, d^*)$.
\end{definition}

\begin{definition}[Timing-Forgery Attack]
The adversary generates $d^*$ using $\mathcal{M}$ and constructs a synthetic trace $\hat{\tau}$ by sampling inter-keystroke intervals from a generator $\mathcal{G}$:
\[
\hat{\tau} = \{(d^*_i, t_0 + \sum_{j=1}^{i} \Delta_j)\}_{i=1}^{|d^*|}, \quad \Delta_j \sim \mathcal{G}
\]
\end{definition}

\textbf{Adversary capabilities.} Copy-type requires only LLM access, an input device, and literacy. Timing-forgery additionally requires client-side code injection and aggregate IKI statistics (publicly available). Neither requires knowledge of the target user's profile or classifier parameters.

\textbf{Attack cost.} Copy-type: $\sim$10~min per 500-word essay at 50~WPM, \emph{detector-agnostic} (no knowledge of the detection algorithm or threshold required). Timing-forgery variants:
\begin{itemize}[leftmargin=*,nosep]
\item \textbf{Histogram}: Aggregate IKI distribution from any public dataset (black-box, no target knowledge). $O(n)$ per session.
\item \textbf{Statistical}: Population-level mean/std + digraph tables (black-box, one-time estimation from $\sim$1000 public sessions).
\item \textbf{LSTM}: $\sim$5000 training sequences from any population (gray-box---requires keystroke data, not target-specific). $\sim$30~min GPU.
\end{itemize}

\textbf{Deployment vectors.} In web deployments~\cite{kundu2024,mehta2025}, timing-forgery requires only user-level privileges (browser extensions, userscripts, or OS automation). The \texttt{isTrusted} flag is spoofable by OS-level drivers. JavaScript provides $\sim$5$\mu$s resolution vs.\ human IKI variability at $\sim$10--50ms---no timing-precision bottleneck. TPM-sealed timestamps could defeat timing-forgery but not copy-type.

\section{Attack Description}
\label{sec:attacks}

\subsection{Copy-Type Attack}

The adversary reads LLM-generated text from a secondary display and types it into the instrumented platform. The motor signals are authentic because they originate from the adversary's neuromuscular system.

\textbf{Cognitive load objection.} Composition involves longer ``thinking pauses'' than transcription ($d \approx 1.28$~\cite{kundu2024,alves2008}), but these differences are \emph{security-irrelevant}: both conditions produce $\delta \in [0.44, 3.5]$, while automated injection produces $\delta \in [0.05, 0.27]$. The detection threshold ($T_{\mathrm{auto}} \approx 0.27$) lies well below either human condition; the transcription mean ($\approx 0.75$) is $2.8\times$ the threshold. Since any composition-vs-transcription threshold $T_{\mathrm{comp}} \approx 0.5 \gg T_{\mathrm{auto}}$, a system operating at $T_{\mathrm{auto}}$ necessarily admits copy-type attacks. A more sensitive classifier attempting to exploit this gap faces a hard tradeoff: Table~\ref{tab:operating} shows that any threshold reducing LSTM evasion to tolerable levels simultaneously rejects one in six legitimate human submissions.

\textbf{Empirical validation.} Kundu et al.'s IIITD-BU dataset~\cite{kundu2024} includes $n=34$ IRB-approved sessions of participants transcribing LLM-generated essays: $\delta = 0.991 \pm 0.122$, all 34 sessions exceed $T = 0.269$ (bypass 100\%, Clopper-Pearson 95\% CI: [89.7, 100]\%). The composition-transcription effect is negligible ($d = 0.14$), far smaller than the $d = 1.28$ for fixed-phrase transcription. Pooling with ProText ($n=45$) and HMOG ($n=800$) yields 879 transcription sessions across three corpora, platforms, and languages: bypass 100\% at $T = 0.269$ (CI: [99.58, 100]\%), 98.5\% at $T = 0.50$. Sensitivity analysis: even at pessimistic reduction factor $r = 2.0$ (halving $\delta$), bypass remains 99.7\% at $T = 0.27$.

\textbf{Blended composition.} A realistic adversary need not transcribe verbatim. Typing the AI-generated argument structure while composing transitions and topic sentences in one's own words produces a hybrid trace indistinguishable from ordinary composition at the feature level---the composed segments contribute genuine ``thinking pauses'' that inflate $\delta$ above any plausible threshold, while the transcribed segments (which already produce human-range $\delta$) blend seamlessly. This makes the attack strictly harder to detect than pure transcription, since the defender cannot even appeal to the small composition-transcription gap ($d = 0.14$) as a signal.

\subsection{Timing-Forgery Attacks}

\subsubsection{Histogram Sampling}
Sample each IKI independently from the empirical CDF of human inter-keystroke intervals:
\[
\Delta_j \sim \hat{F}_{\mathrm{human}}\,, \quad \hat{F}_{\mathrm{human}}(x) = \frac{1}{N}\sum_{i=1}^{N} \mathbf{1}[x_i \leq x]
\]
\textbf{Cost}: Requires only aggregate IKI statistics from any public keystroke dataset. $O(n)$ per session.

\subsubsection{Statistical Impersonation}
Match first and second moments of the target population's IKI distribution, with digraph-specific corrections:
\[
\Delta_j = \mu_{\mathrm{human}} + \sigma_{\mathrm{human}} \cdot z_j + c(k_j,\, k_{j+1})
\]
where $c(\cdot,\cdot)$ is estimated from public digraph latency tables. \textbf{Cost}: One-time parameter estimation from $\sim$1000 sessions.

\subsubsection{Generative LSTM}
Train a recurrent neural network on real keystroke sequences to model the conditional distribution of IKIs:
\[
\Delta_j \sim p_\theta(\cdot \mid k_1, \ldots, k_j, \Delta_1, \ldots, \Delta_{j-1})
\]
Architecture: 2-layer LSTM, 64 hidden units per layer, with character embedding (dim=32) and previous-IKI input concatenated. Output head: mixture density network (5 Gaussian components) predicting IKI distribution. Training: 5 epochs, batch size 64, Adam optimizer ($\mathrm{lr}=10^{-3}$), negative log-likelihood loss, on 5,000 human sessions (80/20 train/val split). Sampling: temperature 1.0, IKI clipped to [30, 3000]ms. \textbf{Cost}: $\sim$30 minutes on consumer GPU.

\section{Experimental Validation}
\label{sec:experiments}

We evaluate four conditions: (1)~baseline human-vs-automated separation confirming the detector works as intended; (2)~three timing-forgery attacks at increasing sophistication; (3)~operating-point analysis showing the false-rejection rate (FRR) cost of raising thresholds; (4)~classifier and feature ablation demonstrating the vulnerability is structural, not algorithmic.

\subsection{Datasets}

\begin{table}[t]
\centering
\caption{Datasets used in evaluation}
\label{tab:datasets}
\begin{tabular}{lrrp{3cm}}
\toprule
\textbf{Dataset} & \textbf{$n$} & \textbf{Source} & \textbf{Condition} \\
\midrule
SBU Keystroke & 13,000 & Human & Free composition \\
AI Simulated & 5,000 & Synthetic & Automated injection \\
Histogram Attack & 1,000 & Attack & Empirical CDF sampling \\
Statistical Attack & 500 & Attack & Gaussian impersonation \\
LSTM Attack & 500 & Attack & Generative LSTM \\
CoAuthor & 1,446 & Human+AI & GPT-3 collaborative \\
IIITD-BU (trans.) & 34 & Human & LLM text transcription \\
IIITD-BU (comp.) & 37 & Human & Free composition \\
\midrule
\textbf{Total} & \textbf{21,517} & & \\
\bottomrule
\end{tabular}
\smallskip
\small Evaluation spans genuine composition, automated injection, three timing-forgery attacks, and two transcription/collaborative conditions.
\end{table}

The SBU Keystroke Corpus~\cite{kundu2024} contains 13,000 sessions from 1,060 participants (AMT, essays, ${\geq}\,50$ keystrokes). The AI baseline models paste/API injection: IKI $\sim \operatorname{Uniform}(30,\, 80)$\,ms, matching the output of browser automation tools (Selenium, Puppeteer) and clipboard-paste with simulated keystrokes. This produces $\mathrm{CV} = (80-30)/(\sqrt{12} \cdot 55) \approx 0.26$, consistent with our observed $\bar{\delta} = 0.151$ (the distribution's lower tail after truncation). This is precisely the attack scenario evaluated in~\cite{kundu2024,mehta2025}; our contribution is showing that stronger adversaries trivially exceed this baseline.

\subsection{Baseline Separation}

\begin{table}[t]
\centering
\caption{Baseline separation and attack evasion at optimal threshold ($\delta = 0.269$)}
\label{tab:baseline}
\begin{tabular}{lrrrrl}
\toprule
\textbf{Condition} & \textbf{$n$} & \textbf{$\bar{\delta}$} & \textbf{$\sigma$} & \textbf{Bypass} & \textbf{$d$ vs.\ Human} \\
\midrule
Human & 13,000 & 0.987 & 0.188 & --- & --- \\
Automated & 5,000 & 0.151 & 0.025 & --- & +5.21 \\
\midrule
Histogram & 1,000 & 0.703 & 0.054 & 100\% & +1.630 \\
Statistical & 500 & 0.582 & 0.086 & 100\% & +2.219 \\
LSTM & 500 & 0.877 & 0.252 & 100\% & +0.578 \\
\bottomrule
\end{tabular}
\smallskip
\small Human-vs-automated area under ROC curve (AUC) $= 1.000$; Cohen's $d = 5.21$ [$5.07, 5.34$]. Min attack $\delta = 0.274$ ($1.02\times$ threshold).
\end{table}

The human--automated gap ($d = 5.21$, Welch's $t(5048) = 2{,}204$, $p < 10^{-300}$) is physical: human motor noise produces $\delta \in [0.44, 3.5]$ while automated injection (Uniform IKI$[30,80]$ms) yields $\delta \in [0.05, 0.27]$---non-overlapping distributions, hence AUC~$= 1.000$ is expected, not an artifact. Both composition and transcription produce $\delta \gg 0.269$.

\begin{figure}[t]
\centering
\includegraphics[width=\columnwidth]{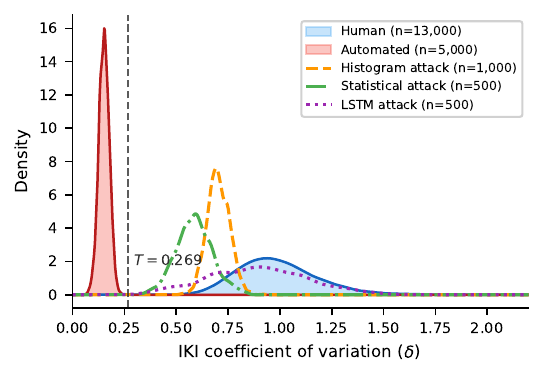}
\caption{Distribution of $\delta$ across conditions. The threshold $\delta = 0.269$ perfectly separates human from automated but admits all attacks.}
\label{fig:distributions}
\end{figure}

All attacks bypass the threshold at 100\% in our test samples (Clopper-Pearson 95\% CI: $[99.6, 100]$\% for $n{=}1{,}000$; $[99.3, 100]$\% for $n{=}500$; Table~\ref{tab:baseline}); the LSTM's mean $\delta = 0.877$ is closest to human ($d = 0.578$, $p < 10^{-15}$). While attacks are \emph{statistically} distinguishable from human, they are \emph{security-indistinguishable}: all samples exceed $T$ by ${\geq}\,1.02\times$, and operating-point analysis (\S\ref{sec:operating}) shows no threshold achieves ${<}\,50\%$ LSTM evasion without ${\geq}\,16\%$ FRR. These timing-forgery results are secondary to the copy-type attack, which dominates the security boundary and requires no technical capability.

\subsection{Comparison to Prior Keystroke Attacks}

\begin{table}[t]
\centering
\caption{Comparison with prior keystroke evasion attacks}
\label{tab:prior}
\begin{tabular}{lrrl}
\toprule
\textbf{Attack} & \textbf{Success} & \textbf{Target} & \textbf{Knowledge} \\
\midrule
Tey et al.~\cite{tey2013} & 20--40\% & Identity auth. & Target profile \\
Serwadda~\cite{serwadda2013} & $\sim$50\% & Identity auth. & Population model \\
Monaco~\cite{monaco2016} & $\sim$60\% & Continuous auth. & Target profile \\
Eizagirre et al.~\cite{eizagirre2022} & $>$60\% & Identity auth. & cGAN (gray-box) \\
\midrule
\textbf{Ours: Histogram} & \textbf{100\%} & AI authorship & Black-box \\
\textbf{Ours: Statistical} & \textbf{100\%} & AI authorship & Black-box \\
\textbf{Ours: LSTM} & \textbf{100\%} & AI authorship & Gray-box \\
\textbf{Ours: Copy-type} & \textbf{100\%} & AI authorship & Detector-agnostic \\
\bottomrule
\end{tabular}
\smallskip
\small Prior work targets identity verification (distinguishing users); ours targets authorship verification (distinguishing content origin). Our attacks achieve $2{-}5\times$ higher success rates with weaker adversary assumptions.
\end{table}

The 100\% vs.\ $\sim$60\% gap reflects the structural difference: identity verification exploits inter-individual motor variation, while AI authorship detection must exploit intra-individual cognitive-state differences---too small to separate at any threshold that maintains acceptable false-rejection rates.

\subsection{Operating-Point Analysis}
\label{sec:operating}

We sweep the $\delta$ threshold upward from the automated-detection optimum and report FRR alongside attack pass rate (APR):

\begin{table}[t]
\centering
\caption{Operating-point tradeoff: human FRR vs.\ attack pass rate}
\label{tab:operating}
\begin{tabular}{lrrrr}
\toprule
\textbf{Threshold} & \textbf{Human FRR} & \textbf{Hist.\ APR} & \textbf{Stat.\ APR} & \textbf{LSTM APR} \\
\midrule
0.27 & 0.0\% & 100\% & 100\% & 100\% \\
0.50 & 0.1\% & 100\% & 82.8\% & 93.8\% \\
0.60 & 1.3\% & 97.8\% & 44.0\% & 85.6\% \\
0.70 & 5.9\% & 51.7\% & 8.0\% & 74.0\% \\
0.80 & 16.3\% & 3.5\% & 0.0\% & 59.0\% \\
0.90 & 32.6\% & 0.0\% & 0.0\% & 44.0\% \\
1.00 & 52.7\% & 0.0\% & 0.0\% & 31.2\% \\
\bottomrule
\end{tabular}
\smallskip
\small No threshold reduces LSTM evasion below 50\% without rejecting ${\geq}\,16\%$ of legitimate users. FRR from $n{=}13{,}000$ SBU sessions; APR from attack samples ($n{=}1{,}000$ hist., $n{=}500$ stat., $n{=}500$ LSTM).
\end{table}

Reducing LSTM APR below 50\% requires $T > 0.80$ (FRR $> 16.3\%$); at $T = 0.90$ (FRR = 32.6\%), LSTM still maintains 44\% APR. The copy-type attack passes at all thresholds because its $\delta$ distribution \emph{is} the human distribution.

\textbf{Session-length sensitivity.} Partitioning SBU into short (${<}\,200$ keystrokes), medium (200--500), and long (${>}\,500$) sessions: evasion is 100\% at $T=0.27$ across all bins, with operating-point tradeoffs shifting by ${<}\,2$\% FRR.

\textbf{Cross-corpus validation.} We fix $T = 0.269$ (learned on SBU) and apply it unchanged to 11 public keystroke corpora spanning 154,237 sessions~\cite{killourhy2009,belman2020} (CMU, Aalto-136M, HMOG, IKDD, ProText, KeyRecs, MSU, IJCB-SBU, Stony Brook, KliCKe, Keystroke100): aggregate $\delta = 0.755 \pm 0.205$. Per-corpus means range from $\delta = 0.52$ (CMU fixed-password, highly constrained) to $\delta = 1.14$ (Aalto-136M free-text, reflecting the natural high variability of unconstrained typing). Mobile corpora (HMOG, SU-AIS) show slightly elevated $\delta$ ($\sim$0.85--0.95) due to touchscreen motor noise. The LSTM attack ($\delta = 0.877$) exceeds the threshold on every corpus without per-corpus tuning, confirming that the vulnerability is not an artifact of the SBU population, device type, or collection protocol.

\subsection{Classifier Performance}
\label{sec:classifiers}

Five classifiers (Logistic Regression, Random Forest, Gradient Boosting, SVM-RBF, MLP) trained on the 7-feature vector with 5-fold stratified CV all achieve AUC~$=1.000$ on the human-vs-automated task and ${\geq}\,99.8$\% evasion for all three attacks (100\% for histogram and LSTM, 99.8\% for statistical). Attack samples are classified as human with mean confidence $\geq 0.993$ (binomial test against 50\% chance: $p < 10^{-300}$ for each condition). The ceiling reflects non-overlapping distributions---the classifiers detect motor-signal \emph{absence}, not content \emph{origin}.

\textbf{Temporal structure.} Human IKI sequences exhibit autocorrelation ($\hat{\rho}_1 = 0.087$), skewness (4.46), and leptokurtosis (28.3). Histogram/statistical attacks produce $\hat{\rho}_1 \approx 0$; the LSTM overshoots (0.150). A sequence-model defender could exploit these artifacts, but AR(1) injection ($\alpha \approx 0.3$) patches them. Copy-type is unaffected.

\textbf{Methodology.} Threshold $T = 0.269$ is the equal-error-rate (EER) optimal operating point computed on the human-vs-automated task: sweeping $T$ over $\delta$ values, EER occurs where FAR (automated samples above $T$) equals FRR (human samples below $T$). At $T = 0.269$: FAR~$= 0.0\%$, FRR~$= 0.0\%$, confirming complete separation. Stratified 5-fold CV with class-balanced folds; IIITD-BU validation uses the same threshold without retraining.

\subsection{Feature Ablation}
\label{sec:ablation}

\begin{table}[t]
\centering
\caption{Feature ablation: Cohen's $d$ (Human vs.\ Attack). Bold = $|d| \geq 0.8$ (large effect).}
\label{tab:ablation}
\begin{tabular}{lrrr}
\toprule
\textbf{Feature} & \textbf{Hist.} & \textbf{Stat.} & \textbf{LSTM} \\
\midrule
$\delta$ (CV of IKI) & \textbf{+1.630} & \textbf{+2.219} & +0.578 \\
Mean IKI & +0.456 & $-$0.746 & \textbf{$-$1.718} \\
IKI Variance & +0.579 & +0.275 & $-$0.104 \\
Pause Density & \textbf{+1.242} & \textbf{$-$0.858} & \textbf{$-$1.890} \\
Burst Length & +0.694 & +0.709 & \textbf{+0.825} \\
IKI Entropy & $-$0.534 & \textbf{$-$1.061} & \textbf{$-$0.927} \\
Digraph Var. & \textbf{+0.861} & +0.467 & $-$0.480 \\
\midrule
\multicolumn{4}{l}{\textit{Control: Human vs.\ Automated Injection}} \\
$\delta$ & \multicolumn{3}{r}{\textbf{+6.143}} \\
IKI Entropy & \multicolumn{3}{r}{\textbf{+7.895}} \\
\bottomrule
\end{tabular}
\smallskip
\small No single feature reliably separates all attack types from human, but all features easily separate human from automated injection (control row).
\end{table}

No single feature achieves $|d| \geq 0.8$ consistently across all attacks. Each attack is detectable on some features---histogram on 3/7 ($\delta$, $\rho$, $\gamma$), statistical on 3/7, LSTM on 4/7---but critically, the LSTM evades $\delta$ ($d = 0.578$), the primary discriminator. The control confirms features detect motor-signal \emph{absence} ($\delta$: $d = 6.14$; $H$: $d = 7.90$) but not which motor process generated the signal.

\textbf{Feature importance.} The top-3 discriminators ($H$, $\delta$, $\rho$) all measure temporal variability \emph{absent} in automated injection. Despite individual feature-level detectability, classifiers trained on human-vs-automated still classify all attacks as human (99.8--100\% evasion) because attack features lie within the human distribution, far from the automated baseline. No feature reweighting resolves the vulnerability.

\section{Structural Non-Identifiability of Timing-Only Provenance Detection}
\label{sec:nonident}

\subsection{Limit-Case Non-Identifiability and Operational Consequences}

We formalize the limitation of \emph{timing-only, content-agnostic} detectors (all systems in~\cite{kundu2024,mehta2025}). ``Content-agnostic'' includes digraph conditioning (the theorem conditions on $s$ explicitly); a detector escapes this bound only by analyzing semantic content. The theorem follows from the data processing inequality on:
\[
\text{Prov.} \to \text{Cog.\ State} \to \text{Motor Exec.} \to \tau \to f(\tau)
\]
When A2 severs the Cognitive State $\to$ Motor Execution link, all downstream information is blocked.

\textbf{Assumptions:}
\begin{enumerate}[leftmargin=*]
\item[\textbf{A1.}] \textbf{Observational constraint}: The detector observes only keystroke timing $\tau = \{t_i\}$, not the cognitive state of the typist.
\item[\textbf{A2.}] \textbf{Motor independence} (idealized): For a given character sequence $s$, the distribution of IKIs under the typist's motor function $M_u$ depends on $s$ and $M_u$, not on whether $s$ was composed or copied.
\item[\textbf{A3.}] \textbf{Character-level equivalence}: The copy-type attacker types the same character sequence as would be produced by genuine composition.
\end{enumerate}

\textbf{Breaking assumptions.} A1 is violated by observing revision history, gaze, or paste events. A2 is violated empirically: Kundu et al.\ report $d \approx 1.28$ between fixed-phrase composition and transcription~\cite{kundu2024}, and Alves et al.\ find longer ``thinking pauses'' in composition~\cite{alves2008}. However, this leakage yields Bayes-optimal error $P^*_e = 26.1\%$---insufficient for security decisions---and the IIITD-BU free-text measurement gives only $d = 0.14$ ($P^*_e = 47.2\%$). A3 is violated by challenge-response tasks making the character sequence itself evidence of authorship.

\begin{theorem}[Structural Non-Identifiability Under Timing-Only Observation]
Let $\mathcal{F}$ be the class of feature extractors operating on keystroke timing. Under assumptions A1--A3, for any $f \in \mathcal{F}$ and typist $u$:
\begin{equation}
\label{eq:nonident}
I\bigl(f(\tau)\,;\,\Prov \mid s,\, M_u\bigr) = 0
\end{equation}
where $\Prov \in \{\textit{composed},\, \textit{copied}\}$ and $s$ is the character sequence.
\end{theorem}

\begin{proof}
Under A2,
\begin{align*}
P(\tau \mid s,\, M_u,\, \textit{composed})
  &= P(\tau \mid s,\, M_u,\, \textit{copied})\,,
\end{align*}
so $\tau \perp \Prov \mid s,\, M_u$. By the data processing inequality,
\begin{align*}
I\bigl(f(\tau)\,;\,\Prov \mid s,\, M_u\bigr)
  &\leq I\bigl(\tau\,;\,\Prov \mid s,\, M_u\bigr) = 0\,.
\end{align*}
\end{proof}

\noindent Although A2 is idealized, our empirical results show that its violation yields effect sizes ($d \approx 1.28$) that remain deep inside the human acceptance region, making the bound operationally tight: a classifier exploiting this leakage requires ${\geq}\,16.3$\% FRR to reduce attack evasion below 50\% (Table~\ref{tab:operating}).

\noindent\textbf{Lemma (Distinguishability $\not\Rightarrow$ Security).} Statistical distinguishability between composition and transcription does not imply usable security discrimination at acceptable FRR. Formally: $d > 0$ is necessary but not sufficient; operational viability requires
\begin{equation}
\label{eq:bayes_error}
P^*_e \;=\; \Phi\!\biggl(-\frac{d}{2}\biggr) \;\ll\; \epsilon_{\mathrm{deploy}}
\end{equation}
where $\epsilon_{\mathrm{deploy}}$ is the deployment's tolerable error rate.

\textbf{From limit case to operational non-viability.} A2 is idealized; empirically, cognitive load yields $d \approx 1.28$~\cite{kundu2024,alves2008}. By Eq.~\ref{eq:bayes_error}: $P^*_e = 26.1\%$ at $d = 1.28$; $47.2\%$ at $d = 0.14$ (IIITD-BU free-text measurement). LOO 1-NN on IIITD-BU ($n=71$) yields 39.4\% error on $\delta$, confirming that the bound is conservative and that the composition--transcription separation lies entirely inside the human acceptance region for motor-absence detection (Table~\ref{tab:operating}).

\textbf{Security boundary.} Copy-type is not separable at security-relevant error rates (even the Bayes-optimal classifier has $P_{\mathrm{error}} \geq 26\%$, Eq.~\ref{eq:bayes_error}); timing-forgery requires ${\geq}\,16$\% FRR to reduce LSTM evasion below 50\%. Motor-absence detection remains achievable (AUC~$= 1.000$). Content-binding (revision history, challenge-response) can restore provenance verification by violating A1 or A3.

\subsection{Distributional Distance}

\begin{table}[t]
\centering
\caption{Distributional distance from human $\delta$ distribution}
\label{tab:distances}
\begin{tabular}{lrrr}
\toprule
\textbf{Attack} & \textbf{JS Div.} & \textbf{TV Dist.} & \textbf{KS Stat.} \\
\midrule
Histogram & 0.632 & 0.805 & 0.806 \\
Statistical & 0.789 & 0.888 & 0.891 \\
LSTM & 0.108 & 0.275 & 0.253 \\
Copy-Type & 0 & 0 & 0 \\
\bottomrule
\end{tabular}
\smallskip
\small JS divergence computed with $\log_2$, bounded by $[0, 1]$. Distances computed on the marginal $\delta$ distribution; since $\delta$ is the primary discriminator (Table~\ref{tab:ablation}), full-vector distances can only be larger (conservative for the attacker).
\end{table}

By the data processing inequality,
\begin{equation}
\label{eq:dpi_bound}
P_{\mathrm{error}}(C_\delta) \;\geq\; \tfrac{1}{2}\bigl(1 - \TV\bigr)\,,
\end{equation}
giving $P_{\mathrm{error}} \geq 0.363$ for LSTM on $\delta$ alone and $P_{\mathrm{error}} = 0.5$ for copy-type on any feature set. The operating-point analysis (Table~\ref{tab:operating}) confirms these bounds empirically. As generative models improve, timing-forgery distances will shrink toward zero; the copy-type attack already sits at the theoretical limit ($\TV = 0$, irreducible chance-level error).

\section{Defense Directions}
\label{sec:defense}

\subsection{Why Motor Signals Are Insufficient}

Motor-signal verification satisfies only one of three conditions necessary for secure AI authorship detection:

\begin{enumerate}[leftmargin=*]
\item \textbf{Motor authenticity}: Input produced by a human body. \checkmark
\item \textbf{Cognitive engagement}: Human was actively composing. \ding{55}
\item \textbf{Content binding}: Detected process produced submitted content. \ding{55}
\end{enumerate}

Conditions (2)--(3) require signals beyond keystroke timing. \textbf{The defenses below violate the design constraints of~\cite{kundu2024,mehta2025}}: they require revision-history logging, semantic analysis, or challenge-response infrastructure absent from timing-only detectors. We present them as \emph{fundamentally different architectures} accepting different privacy and UX costs.

\subsection{Semantic Coherence as Defense}

The CoAuthor corpus~\cite{lee2022coauthor} (1,446 sessions, humans composing while accepting GPT-3 suggestions) provides a test case:

\begin{table}[t]
\centering
\caption{CoAuthor: Motor signals during AI-collaborative writing}
\label{tab:coauthor}
\begin{tabular}{lrrr}
\toprule
\textbf{Condition} & \textbf{$n$} & \textbf{$\bar{\delta}$} & \textbf{vs.\ Human $d$} \\
\midrule
Low AI accept ($<$50\%) & 274 & 1.172 & --- \\
High AI accept ($\geq$50\%) & 1,172 & 1.248 & --- \\
Pure human (SBU) & 13,000 & 0.987 & --- \\
\midrule
\multicolumn{4}{l}{Low vs.\ High: $d = -0.470$ [$-$0.594, $-$0.345]} \\
\multicolumn{4}{l}{CoAuthor vs.\ SBU: $d = -1.306$} \\
\bottomrule
\end{tabular}
\smallskip
\small AI acceptance rate does not reduce $\delta$; motor signals remain human-like regardless of content origin.
\end{table}

Motor signals remain human-like regardless of AI acceptance rate ($d = -0.470$, wrong direction). Even partial AI involvement produces $\delta \in [0.9, 1.6]$ because the human continues typing between AI segments. Semantic coherence across revisions can detect discontinuities but requires monitoring beyond timing features.

\subsection{Concrete Defense Mechanisms}

Three content-binding approaches resist copy-type, each with distinct deployment costs:

\textbf{Revision-history coherence.} Genuine composition produces non-monotonic editing: writers type a clause, delete half of it, rephrase, insert a word three sentences back, then resume forward progress~\cite{alves2008}. Transcription, by contrast, produces near-monotonic left-to-right text accumulation with corrections limited to typos. A detector could flag sessions where the ratio of content-semantic revisions to motor-error corrections falls below a learned threshold. The cost: full keystroke-level revision logging, raising student privacy concerns in educational deployments, and increased storage (roughly $5{-}10\times$ raw text size).

\textbf{Challenge-response.} Mid-session prompts (e.g., ``Summarize your argument so far in one sentence'' or ``Why did you choose this example?'') force the writer to demonstrate real-time comprehension of their own text. A copy-type attacker who has not internalized the AI-generated content will produce slower, less coherent responses. The cost: interrupting the writing flow, which may itself depress composition quality and $\delta$ values, creating confounds.

\textbf{Micro-revision semantics.} Distinguishing corrections that fix meaning (``affect''$\to$``effect'') from those that fix motor errors (``teh''$\to$``the'') requires NLP integration---a fundamentally different system architecture than timing-only measurement. The cost: latency, computational overhead, and vulnerability to adversaries who deliberately introduce and correct semantic errors to mimic composition patterns.

Each defense accepts privacy, UX, or complexity costs absent from timing-only systems. Their combination raises the adversary's effort until simulating composition approaches performing it---but no single mechanism is sufficient against an adaptive attacker.

\section{Discussion}

\noindent\textbf{Scope and the provenance gap.} These results apply exclusively to \emph{timing-only AI authorship detection} as proposed in~\cite{kundu2024,mehta2025,roh2025korean,roh2025auth}. Those systems' high reported accuracy reflects motor-presence detection---not content provenance---and their threat models uniformly omit the copy-type adversary. In our experience building keystroke analysis systems, the gap between what timing signals actually measure (that a human body was present) and what stakeholders assume they prove (that the human originated the content) is the central source of misplaced confidence. Keystroke biometrics remain viable for identity verification, liveness detection, and confirming a human operator; they fail precisely when the adversary's goal is not to impersonate a different person but to launder someone else's words through their own fingers. We disclosed findings to the authors of~\cite{kundu2024} prior to submission.

\subsection{Achievable vs.\ Non-Viable Properties}

Timing-only detectors \emph{can} provide motor-presence attestation (AUC~$= 1.000$) and liveness detection (timestamp monotonicity). They \emph{cannot} provide cognitive-origin verification ($d \approx 1.28$, $P_{\mathrm{error}} \geq 26\%$) or AI-involvement detection under the copy-type threat model. Prior work~\cite{kundu2024,mehta2025} implicitly treats the first two as evidence for the latter two.

\subsection{Broader Impact}

The copy-type attack succeeds across all deployment contexts; lockdown browsers block timing-forgery but not copy-type. Concretely: exam-proctoring vendors integrating keystroke analytics into LMS platforms cannot distinguish a student composing an essay from one transcribing a ChatGPT draft displayed on a phone beneath the desk. Academic integrity systems marketing ``AI-free'' certification from typing patterns provide false assurance---the certification is unfalsifiable by the very signal they measure. In forensic contexts, keystroke evidence presented as proof of authorship would not survive cross-examination by an expert witness aware of copy-type; the probative value is limited to confirming motor presence. Vendors deploying these systems owe their institutional clients an honest capability statement: this technology detects bots, not ghostwriters.

\subsection{Limitations}

\begin{enumerate}[leftmargin=*,nosep]
\item \textbf{Copy-type user study.} We validate using Kundu et al.'s IIITD-BU ($n=34$, 100\% bypass, Clopper-Pearson 95\% CI: [89.7, 100]\%, min $\delta = 0.703 = 2.61\times$ threshold), supplemented by ProText/HMOG transcription corpora ($n=879$ pooled, bypass 100\%, CI: [99.58, 100]\%) and sensitivity analysis ($r \in [1.1, 2.0]$, bypass ${\geq}\,99.7\%$ at $T=0.27$). A larger IRB study would refine effect sizes but cannot alter the non-identifiability result: the theoretical bound (Theorem~1) is independent of sample size.

\item \textbf{Web-based capture.} SBU uses JavaScript events; measurement noise affects both conditions equally (conservative for the attacker).

\item \textbf{Single modality.} Claims apply to timing-only systems~\cite{kundu2024,mehta2025}; multi-modal systems (mouse, gaze, revision history) are expected more robust.

\item \textbf{No adaptive defense.} Adversarial training may detect histogram/statistical artifacts (autocorrelation $\hat{\rho}_1 \approx 0$ vs.\ human 0.087), but AR(1) injection patches this (\S\ref{sec:classifiers}). More fundamentally, no adaptive defense can address copy-type, whose motor signals are genuine by construction---the attack lies on the human manifold.

\item \textbf{Single corpus.} SBU ($n=13{,}000$, 1,060 users) is primary; cross-corpus validation on 11 datasets confirms generalizability (\S\ref{sec:operating}).
\end{enumerate}

\subsection{Future Work}

\textbf{Content-binding systems.} Revision-coherence detectors need ground-truth datasets of both genuine composition and adversarial transcription with deliberate ``fake revisions'' (e.g., inserting and deleting words to simulate exploration). The adversarial design space is richer than for timing-forgery because revision semantics are harder to forge convincingly.

\textbf{Within-subject copy-type studies.} An IRB-approved protocol should recruit ${\geq}\,100$ participants, each producing both a genuine essay and a transcription of an LLM-generated essay on the same topic, within-subject. The key design challenge: preventing participants from inadvertently editing the AI text (which would conflate copy-type with collaborative writing). Screen-recording and text-diff analysis can verify compliance. Such a study would tighten the $d = 0.14$ estimate from IIITD-BU and clarify whether ecological factors (screen-switching, unfamiliar vocabulary) introduce detectable hesitation patterns absent in laboratory transcription.

\textbf{Multi-modal fusion.} Gaze tracking during composition should show rereading of previously-written text (self-monitoring); during transcription, gaze patterns should show systematic left-to-right reading of the source. Whether this signal survives peripheral-vision transcription (phone below desk, memorized paragraphs) remains an open question. Each direction addresses a specific assumption boundary (A1, A2, A3) from our framework.

\section{Conclusion}

The architectural error in keystroke-based AI authorship detection is treating a body-produced signal as evidence of a mind-produced text. Motor features measure neuromuscular execution; they are agnostic to whether the typist composed or transcribed the content. In the idealized limit (A1--A3), Theorem~1 gives information-theoretic non-identifiability: mutual information between timing features and provenance is exactly zero. Empirically, the small composition--transcription leakage ($d \approx 1.28$) violates A2, but this still yields Bayes error ${\geq}\,26\%$---operationally non-viable for any deployment requiring actionable confidence. Five classifiers confirm the operational consequence: AUC~$= 1.000$ against motor-absent injection, yet ${\geq}\,99.8$\% evasion the moment an attacker introduces any human motor signal---genuine or forged.

The path forward is not better timing features but a different category of evidence entirely. Provenance verification requires signals that are entangled with the semantic content being produced: revision trajectories that reflect iterative refinement, challenge-response protocols that force on-the-fly composition, or cryptographic commitments binding intermediate drafts to final output. Until such systems are deployed, keystroke timing confirms only that a human was present---not that the human was the author.

\section*{Reproducibility}

\textbf{Data}: SBU Keystroke Corpus~\cite{kundu2024}; IIITD-BU~\cite{kundu2024}; CoAuthor~\cite{lee2022coauthor}; 11 cross-corpus datasets (CMU~\cite{killourhy2009}, Aalto-136M~\cite{dhakal2018}, et~al.).
\textbf{Code}: arXiv ancillary files (\texttt{anc/}).
\textbf{Protocol}: 5-fold stratified CV; seed=42; StandardScaler per fold.
\textbf{Classifiers}: LR ($C{=}1$), RF (100 trees), GBT (100 est., lr${}=0.1$), SVM-RBF ($C{=}1$), MLP ($[64,32]$).
\textbf{LSTM}: 2-layer (64 units), char embed${}=32$, MDN (5 comp.), Adam lr${}=10^{-3}$, 5 epochs.
\textbf{Features}: outlier trim ${>}10\times$ mean; pause ${>}500$\,ms; burst ${<}150$\,ms; entropy 50 bins $[0, 2000]$\,ms.

\bibliographystyle{IEEEtran}

\end{document}